\pdfoutput=1 
\documentclass[a4paper,english,cleveref,autoref,thm-restate]{lipics-v2019}

\usepackage{physics}
\usepackage{amsmath}
\usepackage{tikz}
\usepackage{mathdots}
\usepackage{yhmath}
\usepackage{cancel}
\usepackage{color}
\usepackage{siunitx}
\usepackage{array}
\usepackage{multirow}
\usepackage{amssymb}
\usepackage{gensymb}
\usepackage{tabularx}
\usepackage{booktabs}
\usetikzlibrary{fadings}
\usetikzlibrary{patterns}
\usetikzlibrary{shadows.blur}
\usetikzlibrary{shapes}
\usepackage{subcaption}

\newcommand{\hide}[1]{ }

\newcommand{\ant}{\mathsf{ant}}

\newcommand{\symmetry}{\textsc{symmetry}}

\newcommand{\cograph}{\textsc{cograph}}

\newcommand{\distHereditary}{\textsc{dist-hereditary}}
\newcommand{\permutation}{\textsc{permutation}}

\newcommand{\id}{\ensuremath{\mathsf{id}}}

\newcommand{\prot}{\ensuremath{\mathsf{prot}}}

\newcommand{\dM}{\ensuremath{\mathsf{dM}}}
\newcommand{\dAM}{\ensuremath{\mathsf{dAM}}}
\newcommand{\dMA}{\ensuremath{\mathsf{dMA}}}
\newcommand{\dMAM}{\ensuremath{\mathsf{dMAM}}}
\newcommand{\dAMA}{\ensuremath{\mathsf{dAMA}}}

\newcommand{\cL}{\mathcal{L}}
\newcommand{\cO}{\mathcal{O}}

\newcommand{\config}{\langle G, \id, I\rangle}

\DeclareMathSymbol{\shortminus}{\mathbin}{AMSa}{"39}

\newcommand{\priv}{\mathrm{priv}}
\newcommand{\pub}{\mathrm{pub}}

\newcommand{\poly}{\mathrm{poly}}

\newcommand{\eps}{\varepsilon}
\newcommand{\polylog}{\mathrm{poly}\log}

\newcommand{\soundness}{\textbf{\textsf{Soundness. }}}
\newcommand{\completeness}{\textbf{\textsf{Completeness. }}}

\newcommand{\lipics}[1]{\textcolor{lipicsGray}{\textbf{\textsf{#1}}}}

\title{Compact Distributed Interactive Proofs for the Recognition of Cographs and Distance-Hereditary Graphs}%TODO Please add

\titlerunning{Distributed Interactive Proofs for Cographs and Distance-Hereditary Graphs} %TODO optional, please use if title is longer than one line

\author{Pedro Montealegre\footnote{Corresponding author}}{Facultad de Ingenier\'{\i}a y Ciencias, Universidad Adolfo Ib\'a\~nez, Chile}{p.montealegre@uai.cl}{}{}
\author{Diego Ram\'{\i}rez-Romero}{Departamento de Ingenier\'{\i}a Matem\'atica, Universidad de Chile, Chile}{dramirez@dim.uchile.cl}{}{}
\author{Ivan Rapaport}{DIM-CMM (UMI 2807 CNRS), Universidad de Chile, Chile}{rapaport@dim.uchile.cl}{}{}

\authorrunning{P. Montealegre, D. Ram\'{\i}rez and I. Rapaport} %TODO mandatory. First: Use abbreviated first/middle names. Second (only in severe cases): Use first author plus 'et al.'

\Copyright{P. Montealegre, D. Ram\'{\i}rez and I. Rapaport} %TODO mandatory, please use full first names. LIPIcs license is "CC-BY";  http://creativecommons.org/licenses/by/3.0/

 %TODO mandatory: Please choose ACM 2012 classifications from https://dl.acm.org/ccs/ccs_flat.cfm 

 %TODO mandatory; please add comma-separated list of keywords

\funding{Partially supported by CONICYT via PIA/ Apoyo a Centros Cient\'{\i}ficos y Tecnol\'ogicos de Excelencia AFB 170001 (P.M. and I.R.), FONDECYT 1170021 (D.R. and I.R.) and FONDECYT 11190482 (P.M.) and PAI + Convocatoria Nacional Subvenci\'on a la Incorporaci\'on en la Academia A\~no 2017 + PAI77170068 (P.M.).}

\nolinenumbers %uncomment to disable line numbering

\hideLIPIcs  %uncomment to remove references to LIPIcs series (logo, DOI, ...), e.g. when preparing a pre-final version to be uploaded to arXiv or another public repository

\EventEditors{}
\EventNoEds{2}
\EventLongTitle{24th International Conference of Principles of Distributed Systems (OPODIS 2020)}
\EventShortTitle{OPODIS 2020}
\EventAcronym{OPODIS}
\EventYear{2020}
\EventDate{December 14--16, 2020}
\EventLocation{Strasbourg}
\EventLogo{}
\SeriesVolume{}
\ArticleNo{0}
\begin{document}
\maketitle

\begin{abstract}

%In this paper, 
We present compact distributed interactive proofs for the recognition of two important graph classes, well-studied in the context of centralized algorithms, namely \emph{complement reducible graphs} and \emph{distance-hereditary graphs}. Complement reducible graphs (also called \emph{cographs}) are defined as the graphs not containing a four-node path $P_4$ as an induced subgraph. Distance-hereditary graphs are a super-class of cographs, defined as the graphs where the distance (shortest paths) between any pair of vertices is the same on every  induced connected subgraph. 

First, we show that there exists a distributed interactive proof for the recognition of cographs with two rounds of interaction. More precisely, 
we give a \dAM\ protocol with a proof size of $\cO(\log n)$ bits that uses shared randomness and recognizes cographs with high probability. Moreover, our protocol can be adapted to verify any Turing-decidable predicate restricted to cographs in $\dAM$ with certificates of size $\cO(\log n)$.

Second, we give a three-round,
\dMAM\ interactive protocol for the recognition of distance-hereditary graphs, still with a proof size of $\cO(\log n)$ bits and also using shared randomness. 

Finally, we show that any one-round (denoted \dM) or two-round, \dMA\ protocol for the recognition of cographs or distance-hereditary graphs requires certificates of size $\Omega(\log n)$ bits. Moreover, we show that any constant-round \dAM\ protocol using shared randomness requires certificates of size $\Omega(\log \log n)$.

\end{abstract}

\newpage

\section{Introduction}\label{sec:introduction}

The study of graph classes provides important insights to address basic graph problems such as coloring, maximum independent set,  
dominating set, etc. Indeed, as such problems are hard in general, it is extremely important to identify structural properties 
of specific instances which can be exploited in order to design efficient algorithms. 

A well-known example is the class of perfect graphs \cite{brandstadt1999graph,10.5555/984029}, that is to say, the class of graphs satisfying that the chromatic number equals the size of the largest clique of every induced subgraph. Many NP-complete problems on general graphs, such as coloring, maximum clique, and maximum independent set,  can be solved in polynomial-time when the input is  known to be a perfect graph \cite{Grtschel1993}. It is therefore very important to efficiently \emph{check the membership} of a graph to a given class. Through this checking procedure  we make sure that the  execution is performed in the right type of input, avoiding  erroneous computations.
 
 In this work we are interested in two subclasses of perfect graphs, namely \emph{complement reducible graphs} and \emph{distance-hereditary graphs}. The class of complement reducible graphs, or simply \emph{cographs}, has several equivalent definitions, as it has been re-discovered in many different contexts \cite{Corneil1981,Jung1978,Seinsche1974, Sumner1974}. A graph is a cograph if it does not contain a four-node path $P_4$ as an induced subgraph. Equivalently, a graph is a cograph if it can be generated recursively from a single vertex by complementation and disjoint-union.  A graph is a \emph{distance-hereditary graph} if the distance between any two vertices is the same on every connected induced subgraph \cite{HOWORKA1977}. An equivalent definition is that every path between two vertices is a shortest path. It is known that every cograph is a distance-hereditary graph. 

Many NP-complete problems are solvable in polynomial-time, or even linear time, when restricted to cographs and distance-hereditary graphs. For instance, maximum clique, maximum independent set, coloring (as distance-hereditary graphs are perfect \cite{HOWORKA1977}), hamiltonicity~\cite{Hung2005}, Steiner tree and connected domination~\cite{DAtri1988}, computing the tree-width and minimum fill-in~\cite{Broersma2000}, among others. By a result of  Courcelle, Makowsky and Rotics~\cite{Courcelle2000}, every decision problem expressible in a type of monadic second order logic can be solved in linear time on distance-hereditary graphs. Observe that all these results also apply to cographs. Other problems, like  graph isomorphism, can be solved in linear time on cographs~\cite{Corneil1981}.

In the centralized setting, both cographs and distance-hereditary graphs can be recognized in linear time \cite{Damiand2001}.  They have also  been addressed in the context of other computation models such as parallel random access machine \cite{dahlhaus1995efficient} or the broadcast congested clique model \cite{kari2015solving,montealegre2020graph} (see the related work section for more details). 

In this work, we focus on the recognition of these classes in the model of distributed interactive proofs. We show that both classes can be recognized with compact certificates and constant (two or three) rounds of interaction. Moreover, our protocol for cographs allows an algorithm to decide any Turing-decidable predicate restricted to cographs in $\dAM[\log n]$.

\subsection{Distributed Interactive Proofs}

Distributed decision refers to the task in which the nodes of a connected graph $G$ have to collectively decide (whether $G$ satisfies) some graph property~\cite{naor1995can}. For performing any such task,
the nodes exchange messages through the edges of $G$. The input of distributed decision problems may also include labels given to the nodes and/or to the edges of $G$. For instance, the nodes could decide whether $G$ is properly colored, or decide whether the graph belongs to a given graph-class.
 
Acceptance and rejection are defined as follows. 
If $G$ satisfies the property, then all nodes must accept; otherwise, at least one node must reject.
This type of algorithm could be used in distributed fault-tolerant computing, where the nodes, with some regularity,
must check whether the current network configuration is in a legal state for some Boolean predicate~\cite{korman2010proof}. Then, if the configuration becomes illegal at some point, the rejecting node(s) raise the alarm or launch a recovery procedure.

Deciding whether a given coloring is proper can be done locally, by exchanging messages between neighbors. These types of properties
are called {\emph{locally decidable}}. Nevertheless, some other properties, such as deciding whether $G$ is a simple path, are not. In fact, this simple observation implies that the recognition of a distance-hereditary graph is not locally-decidable. 
 As a remedy, the notion of {\emph{proof-labeling scheme}} (PLS) was introduced~\cite{korman2010proof}. Similar variants were also introduced: non-deterministic local decisions~\cite{fraigniaud2013towards}, locally checkable proofs~\cite{goos2016locally}, and others.
 
 Roughly speaking, in all these models, a powerful prover gives to every node $v$ a certificate $c(v)$.
This provides $G$ with a global distributed-proof. Then, every node $v$ performs a
local verification using its local information together with $c(v)$. PLS can be seen as a distributed counterpart to the class NP,
where, thanks to nondeterminism, the power of distributed algorithms increases.

%In fact, with the rise of the Internet, prover-assisted computing models are more relevant than ever. We can think of asymmetric applications like Facebook, where, together with the social network itself, there is a very powerful central entity that stores a large amount of data (the topology of the network, preferences, and activities of the users, etc.).  Or we can consider Cloud Computing, where computationally limited devices delegate costly computations to a cloud with tremendous computational power. The central point lies in the fact that these devices may not trust their cloud service (as it may be malicious, selfish, or buggy). Therefore, the nodes must regularly verify the correctness of the computation performed by the cloud service.

	Just as it happened in the centralized framework~\cite{goldreich1991proofs,goldwasser1989knowledge},
a natural step forward is to consider a model where the nodes are allowed to have {\emph{more than one interaction round}} with the prover. Interestingly, there is no gain when interactions are all deterministic. When there is no randomness, the prover, from the very beginning, has all the information required to simulate the interaction with the nodes. Then, in just one round, the prover could simply send to each node the transcript of the whole communication, and the nodes simply verify that the transcript is indeed consistent. A completely different situation occurs when the nodes have access to some kind of randomness~\cite{baruch2015randomized,fraigniaud2019randomized}. In that case, the exact interaction with the nodes is unknown to the prover until the nodes communicate the realization of their random variables. Adding a randomized phase to the non-deterministic phase gives more power to the model~\cite{baruch2015randomized, fraigniaud2019randomized}.
The notion of \emph{distributed interactive protocols} was introduced by Kol, Oshman, and Saxena in~\cite{kol2018interactive} and further studied in~\cite{crescenzi2019trade,fraigniaud2019distributed, montealegre2020shared,naor2020power}. In such protocols, a centralized non- trustable prover with unlimited computation power, named Merlin, exchanges messages with a randomized distributed algorithm, named Arthur. Specifically, Arthur and Merlin perform a sequence of exchanges during which every node queries Merlin by sending a random bit-string, and Merlin replies to each node by sending a bit-string called proof. Neither the random strings nor the proofs need to be the same for each node. After a certain number of rounds, every node exchanges information with its neighbors in the network, and decides (i.e., it outputs accept or reject).  For instance, a \dMAM\ protocol involves three interactions: Merlin provides a certificate to Arthur, then Arthur queries Merlin by sending a random string. Finally, Merlin replies to the query by sending another certificate. Recall that this series of interactions is followed by a phase of distributed verification performed between every node and its neighbors.

When the number of interactions is $k$ we refer to $\dAM[k]$ protocols (if the last player is Merlin) and $\dMA[k]$ protocols (otherwise). For instance, $\dAM[2] = \dAM$, $\dMA[3] = \dAMA$, etc. Also, the scenario of distributed verification, where there is no randomness and only Merlin interacts, corresponds to $\dAM[1]$, which we denote by \dM. In other words, \dM\ is the PLS model.
In distributed interactive proofs, Merlin tries to convince the nodes that $G$ satisfies some property in a small number of rounds and through short messages. We say that an algorithm uses $\cO(f(n))$ bits if the messages exchanged between the nodes (in the verification round) and also the messages exchanged between the nodes and the prover are upper bounded by $\cO(f(n))$. We include this {\emph{bandwidth bound}} in the notation, which becomes $\dMA[k,f(n)] $ and $\dAM[k,f(n)]$
for the corresponding protocols.

	It is known that all Turing-decidable predicates on graphs admit a proof-labeling scheme with certificates of size $\cO(n^2)$ bits. Interestingly, some distributed problems are hard, even when a powerful prover provides the nodes with certificates.
It is the case of \symmetry, the language of graphs having a non-trivial automorphism (i.e., a non-trivial one-to-one mapping from the set of nodes to itself preserving edges). Any 
proof-labeling scheme recognizing \symmetry\ requires certificates of size $\Omega(n^2)$~\cite{goos2016locally}. However, many problems requiring $\Omega(n^2)$-bit certificates in any PLS, such as \symmetry, admit distributed interactive protocols with small certificates, and very few interactions. In fact, 
$\symmetry$ is in both $\dMAM[\log n]$ and $\dAM[n \log n]$ using shared randomness~\cite{kol2018interactive}, with the former result being tight for 
that type of randomness~\cite{montealegre2020shared}.

\subsection{Our Results}
In Section~\ref{sec:cograph} we show that the recognition of cographs is in $\dAM[\log n]$ using shared randomness. Our result consists of adapting an algorithm given in \cite{kari2015solving,montealegre2020graph}, originally designed for the Broadcast Congested Clique model. In this regard, we exploit the natural high connectivity of this class, combined with the use of non-determinism in order to route all messages in the network to a leader node, which is delegated to act as a \emph{referee}. In fact, our protocol allows this leader to learn all the edges of the input graph. We use this fact to show that  any Turing-decidable predicate restricted to cographs is decidable in $\dAM[\log n]$ using shared randomness. 
Then, in Section~\ref{sec:dist_hed}, we adapt the protocol described in the previous section, and we combine it with a set of tools related to the structure of distance-hereditary graphs 
in order to show that the recognition of this class is in $\dMAM^{\pub}[\log n]$. In this case, we are not able to show that all the information can be gathered in a single node representing the referee. Instead, we find a way to verify each step of the computation of the referee in a distributed manner, by choosing nodes that can receive (with the help of the prover) all necessary messages for performing the task. 
Finally, in Section~\ref{sec:lower}, we show that any \dM\ or \dMA\ protocol for the recognition of cographs or distance-hereditary graphs requires messages of size at least $\Omega(\log n) $. Our results are obtained extending a lower-bound technique described in~\cite{goos2016locally}, for the detection of a single leader in the context of locally checkable proofs.

Interestingly, all the protocols given in previous sections use shared randomness. We show that any constant-round $\dAM$ protocol using shared randomness for the previous problems requires messages of size at least $\Omega(\log\log n)$.

\subsection{Related Work}

The recognition of cographs and distance-hereditary graphs has been studied thoroughly in the parallel setting, where both problems have been shown to be in \textsf{NC}~\cite{dahlhaus1995efficient, kirkpatrick1990parallel, he1993parallel}. The currently best algorithms for the recognition of both classes run in time $\cO(\log^2 n)$ and using a linear number of processors  in a CREW-PRAM \cite{dahlhaus1995efficient}. 
On the other hand, there exist fast-parallel algorithms for NP-hard problems restricted to cographs and distance-hereditary graphs  \cite{HSIEH1999, Lin1990}.

Both recognition problems have also been addressed in the \textsf{One-Round  Broadcast Congested Clique Model} (\textsf{1BCC}), also known as the \textsf{Distributed Sketching Model} \cite{DBLP:conf/podc/AssadiKO20}. In this model, the nodes of a graph send a single message to a \emph{referee}, which initially has no information about the graph and, only using the received messages, has to decide a predicate of the input graph. 
	In \cite{kari2015solving}, a public-coin randomized protocol for recognizing cographs is obtained. In~\cite{montealegre2020graph}, randomized protocols recognizing both classes of  graphs  are given. Interestingly, these  protocols not only recognize the classes but \emph{reconstruct them}, meaning that the referee learns all the edges of the input graph. Also, the structural properties of distance-hereditary graphs have been used in the design of compact routing tables for 
interconnection networks~\cite{cicerone2001compact}.
Regarding local certification, other results on the recognition of graph classes include planar graphs~\cite{feuilloley2020compact} and graphs with bounded genus~\cite{feuilloley2020local}, where the authors showed that both classes admit proof-labeling schemes with logarithmic-sized certificates. 
Recently, Naor, Parter and Yogev defined in \cite{naor2020power} a \emph{compiler} which \lipics{(1)} turns any problem solved in \textsf{NP} in time $\tau(n)$ into a $\dMAM$ protocol using private randomness and bandwidth $\tau(n) \log n/n$ and; \lipics{(2)} turns any problem which can be solved in \textsf{NC} into a $\dAM$ protocol with private randomness, $\poly\log n$ rounds of interaction and bandwidth $\polylog n$. For example, this result implies that, any class of sparse graphs that can be recognized in linear time, can also be recognized by a $\dMAM$ protocol with logarithmic-sized certificates. 
Observe that, while it is known that cographs and distance-hereditary graphs can be recognized in \textsf{NC}~\cite{dahlhaus1995efficient} and in linear time in the centralized setting~\cite{Damiand2001}, our protocols beat the performance of the compiler. In fact, both graph classes can have 
$\Theta(n^2)$ edges and, therefore, the use of the compiler shows that the recognition of these classes is in $\dMAM[n \log n]$ or in $\dAM[\poly \log n, \poly \log n]$. The protocols given in this work achieve a bandwidth cost of $\cO(\log n)$, using shared randomness and with two and three rounds of interaction respectively.
In \cite{crescenzi2019trade, montealegre2020shared} the role of shared and private randomness in distributed interactive proofs  is studied. In particular, it is shown that  $\dAM$ protocols with private randomness are in general more powerful than those using only a shared coin, up to a constant increase on the error probability and a logarithmic term in the size of the certificates. This is interesting because the recognition protocols given in this article use only shared randomness, i.e. the \emph{weakest form of randomness}. 

\section{Preliminaries}\label{sec:model-def}

We start giving some graph-theoretic background, and then we formally define distributed interactive proofs, together with the problems we intend to solve.

\subsection{Background on Cographs and Distance-Hereditary Graphs}
All the graphs in this paper are simple undirected graphs, which is a pair of finite sets $G =(V, E)$ where $V$ is called vertex set, and $E$ is a subset of the 2-sets of  $V$  called edge set, that is, $E\subseteq {V \choose 2}$. For a set $U\subseteq V$, we define the induced subgraph of $G=(V,E)$ according to $U$ as the pair $H= (U, E(U))$, where $E(U)=E\cap {U \choose 2}$. Whenever such a graph exists we say that $H$ is an induced subgraph of $G$ and denote it by $H\subseteq G$. If, instead, we have a graph with vertex set $U$ such that its edges are only contained in $E(U)$ we simply call it a subgraph of $G$. A spanning subgraph of $G$ is a subgraph $H$ with $V(H)=V(G)$.

A path in a graph $G$ is an ordered collection of nodes $v_1, \dots v_k$ such that for all $i\in {1, \dots k-1}$ the pair $v_i$ and $v_{i+1}$ are adjacent. Similarly, a cycle can be defined as a path where $v_k $ and $v_1$ are also adjacent. We say that a graph $G$ is \emph{connected} if for any pair of vertices $u,v\in G$ there exists a path $\{w_1,\dots w_k\}$ where $w_1 = u$ and $w_k=v$ for some integer $k$.  Given two nodes $u,v$ in a connected graph $G$ the distance between them, denoted by $d(u,v)$ is defined as the length of the shortest path between $u$ and $v$.

A tree $T$ is an undirected graph such that it is connected and does not have any cycles.  A $P_4$ is an induced path of length four. 

Given two graphs $G_1=(V_1, E_1)$ and $G_2=(V_2,E_2)$, we define the \emph{union} between both graphs, denoted by $G_1\cup G_2$ as the graph $\tilde{G}=(\hat{V}, \tilde{E})$, with $\tilde{V} = V_1\cup  V_2$ and $\tilde{E} = E_1\cup E_2 $. Given two graphs $G_1=(V_1, E_1)$ and $G_2=(V_2,E_2)$, we define the \emph{join} between both graphs, denoted by $G_1*G_2$ as the graph $\hat{G}=(\hat{V}, \hat{E})$, with $\hat{V} = V_1\cup  V_2$ and $\hat{E} = E_1\cup E_2 \cup \{v_1v_2 \: \text{ such that } v_1\in V_1, v_2 \in V_2\}$.

 The set of neighbors of a node $u$ is denoted $N(u)$, and the closed neighborhood $N[u]$ is the set $N(u)\cup \{u\}$. A node $v$ is said to be a \emph{pending node} if it has a unique neighbor in the graph. A pair of nodes $u,v\in V$ are said to be \emph{twins} if their neighborhoods are equal.  That is, $N(u) = N(v)$ or $N[u]= N[v]$.  In the case that $u$ and $v$ are adjacent ($N[u]= N[v]$) we refer to them as \emph{true} twins and, otherwise, we refer to them as \emph{false} twins.

As we mentioned in the introduction, a cograph is a graph that does not contain a  $P_4$ as an induced subgraph (i.e. it is  $P_4$-\emph{free}). Another equivalent definition states that cographs are the graphs which can be obtained recursively following three rules: \lipics{(1)} A single vertex is a cograph, \lipics{(2)} the disjoint union between two cographs is a cograph and \lipics{(3)} the join of two cographs is a cograph.
An advantage of cographs is that they admit other characterizations that may be useful for local verification. First, we define a \emph{twin ordering} as an ordering of the nodes $V=\{v_i\}_{i=1}^n$ such that, for each $j\ge2$,  $v_j$ has a twin in $G(v_1, \dots v_j)$.

\begin{proposition}[\cite{kari2015solving}] \label{cographCharac}Given a graph $G$ the following are equivalent:
	
	\begin{enumerate}
		\item $G$ is a cograph.
		\item Each non trivial induced subgraph of  $G$ has a pair of \emph{twins}.
		\item $G$ is $P_4$-\emph{free}.
		\item $G$ admits a \emph{twin ordering}.
	\end{enumerate}
\end{proposition}

\begin{figure}[!ht]
	\centering
	\includegraphics[width=0.4\linewidth]{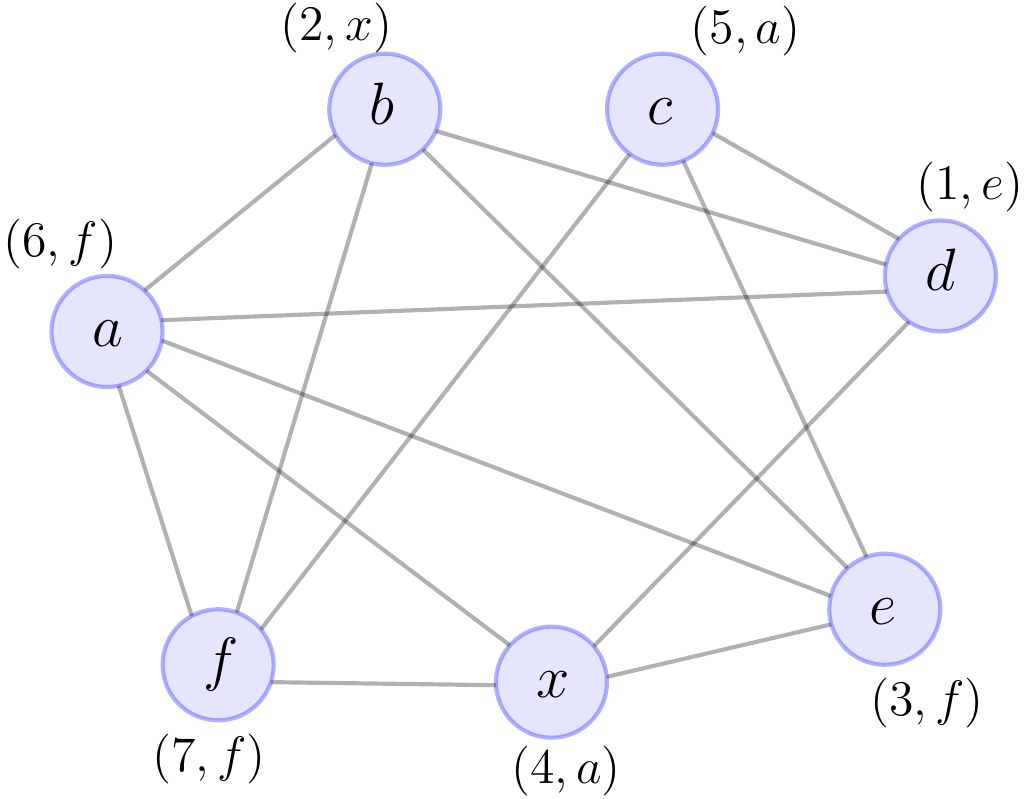}
	\caption{A cograph with labels according to a twin ordering. The first entry represents the step at which they are removed, while the second entry indicates the node's twin at such step.}
	\label{fig:cograph}
\end{figure}

% TODO: \usepackage{graphicx} required
\begin{figure}[!ht]
	\centering
	\includegraphics[width=0.6\linewidth]{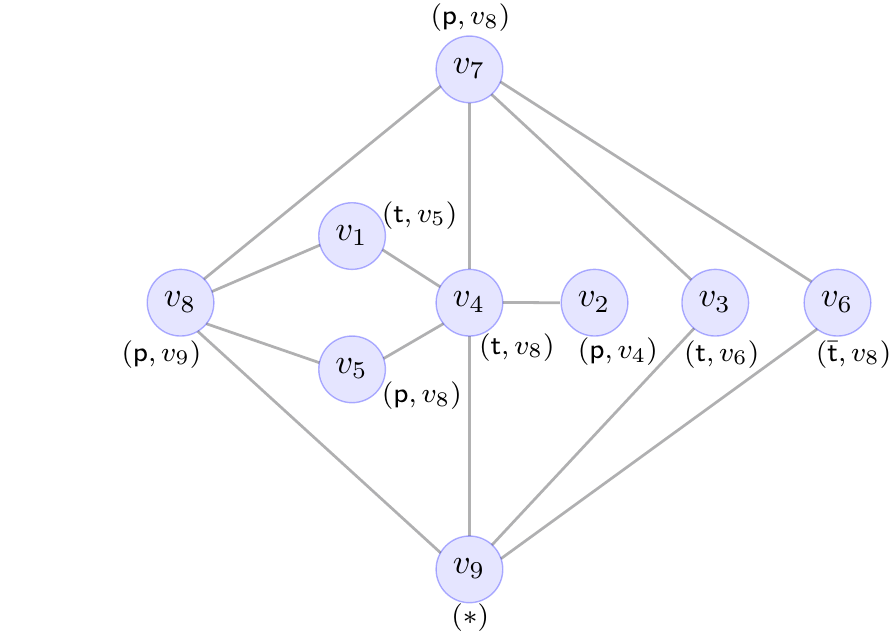}
	\caption{A distance-hereditary graph with labels according to its ordering. The first entry indicates whether it is removed as a \emph{true twin} ($\bar{\textsf{t}}$), a \emph{false twin} (\textsf{t}) or a pending node (\textsf{p}).}
	\label{fig:distheredex}
\end{figure}

A graph $G$  is said to be distance-hereditary if for any induced subgraph $H\subseteq G$ and any pair  $u,v\in H$ satisfy that $d_H(u,v)=d_G(u,v)$. That is, any induced path between a pair of nodes is a shortest path. A relevant characterization for this class is the following.
\begin{proposition}[\cite{brandstadt1999graph}]
	An $n$-node graph $G$ is said to be distance hereditary iff there exists an ordering  $\{v_i\}_{i=1}^n$ such that, for any $i\in [n]$, either there exists $j< i$ such that $v_i$ and $v_j$ are \emph{twins} in $G_i= G(v_1,.. v_i)$ or $v_i$ is a pending node at $G_i$.
\end{proposition}

\subsection{Model Definitions}

Let $G$ be a simple connected $n$-node graph, let $I:V(G)\to \{0,1\}^*$ be an input function assigning labels to the nodes of $G$, where the size of all inputs is polynomially bounded on $n$. Let $\id:V(G)\to\{1,\dots,\text{poly}(n)\}$ be a one-to-one function assigning identifiers to the nodes. A \emph{distributed language} $\mathcal L$ is a (Turing-decidable) collection of triples $(G,\id,I)$, called \emph{network configurations}.

A distributed interactive protocol consists of a constant series of interactions between a {\emph{prover}} called Merlin, and a {\emph{verifier}} called Arthur. The prover Merlin is centralized, has unlimited computing power and knows the complete configuration $(G,\id,I)$. However, he can not be trusted. On the other hand, the verifier Arthur is distributed, represented by the nodes in $G$, and has limited knowledge. In fact, at each node $v$, Arthur is initially aware only of his identity $\id(v)$, and his label $I(v)$. He does not know the exact value of $n$, but he knows that there exists a constant $c$ such that $\id(v) \leq n^c$. Therefore, for instance, if one node $v$ wants to communicate his $\id(v)$ to its neighbors, then the message is of size $\cO(\log n)$.

Given any network configuration $(G, \id, I)$, the nodes of $G$ must collectively decide whether $(G, \id, I)$ belongs to some distributed language ${\mathcal L}$. If this is indeed the case, then all nodes must accept; otherwise, at least one node must reject (with certain probabilities, depending on the precise specifications we are considering).

There are two types of interactive protocols: Arthur-Merlin and Merlin-Arthur. Both types of protocols have two phases: an interactive phase and a verification phase. Let us define first {\emph{Arthur-Merlin interactive protocols}}. If Arthur is the party that starts the interactive phase, he picks a random string $r_1(v)$ at each node $v$ of $ G $ (this string could be either private or shared) and send them to Merlin. Merlin receives $r_1$, the collection of these $n$ strings, and provides every node $v$ with a certificate $c_1(v)$ that is a function of $v$, $r_1$ and $(G,\id,I)$. Then again Arthur picks a random string $r_2(v)$ at each node $v$ of $G$ and sends $r_2$ to Merlin, who, in his turn, provides every node $v$ with a certificate $c_2(v)$ that is a function of $v$, $r_1$, $r_2$ and $(G,\id,I)$. This process continues for a fixed number of rounds. If Merlin is the party that starts the interactive phase, then he provides at the beginning every node $v$ with a certificate $c_0(v)$ that is a function of $v$ and $(G,\id,I)$, and the interactive process continues as explained before. In Arthur-Merlin protocols, the process ends with Merlin. More precisely, in the last, $k$-th round, Merlin provides every node $v$ with a certificate $c_{\lceil k/2\rceil}(v)$. Then, the verification phase begins. This phase is a one-round deterministic algorithm executed at each node. More precisely, every node $v$ broadcasts a message $M_v$ to its neighbors. This message may depend on $\id(v)$, $I(v)$, all random strings generated by Arthur at $v$, and all certificates received by $v$ from Merlin. Finally, based on all the knowledge accumulated by $v$ (i.e., its identity, its input label, the generated random strings, the certificates received from Merlin, and all the messages received from its neighbors), the protocol either accepts or rejects at node $v$. Note that Merlin knows the messages each node broadcasts to its neighbors because there is no randomness in this last verification round.

A {\emph{Merlin-Arthur interactive protocols}} of $k$ interactions is an Arthur-Merlin protocol with $k-1$ interactions, but where the verification round is randomized. More precisely, Arthur is in charge of the $k$-th interaction, which includes the verification algorithm. The protocol ends when Arthur picks a 
random string $r(v)$ at every node $v$ and uses it to perform a (randomized) verification algorithm. In other words, each node $v$ randomly chooses a message $M_v$ from a distribution specified by the protocol, and broadcast $M_v$ to its neighbors. Finally, as explained before, the protocol either accepts or rejects at node $v$. Note that, in this case, Merlin does not know the messages each node broadcasts to its neighbors (because they are randomly generated). If $k=1$, a distributed Merlin-Arthur protocol is a (1-round) randomized decision algorithm; if $k=2$, it can be viewed as the non-deterministic version of randomized decision, etc.

\begin{definition}
Let ${\mathcal V}$ be a verifier and ${\mathcal M}$ a prover of a distributed interactive proof protocol for languages over graphs of $n$ nodes.
If $({\mathcal V}, {\mathcal M})$ corresponds to an Arthur-Merlin (resp. Merlin Arthur) $k$-round, $\cO(f(n))$ bandwidth protocol,
we write $({\mathcal V}, {\mathcal M}) \in {\dAM}_{\prot}[k,f(n)]$ (resp. $({\mathcal V}, {\mathcal M}) \in {\dMA}_{\prot}[k,f(n)]$).
\end{definition}

\begin{definition}
Let $\eps \leq 1/3$. The class $\dAM_{\eps}[k,f(n)]$ (resp. $\dMA_{\eps}[k,f(n)]$) is the class of languages ${\mathcal L}$ over graphs of $n$ nodes
for which there exists a verifier ${\mathcal V}$ such that, for every configuration $(G,\id,I)$ of size $n$, the
two following conditions are satisfied.

\begin{itemize}
\item \emph{\completeness}
If $(G,\id,I) \in \cL$ then, there exists a prover $ {\mathcal M}$ such that 

$({\mathcal V}, {\mathcal M}) \in {\dAM}_{\prot}[k,f(n)]$
(resp. $({\mathcal V}, {\mathcal M}) \in {\dMA}_{\prot}[k,f(n)]$) and 
\noindent
\[\mathbf{Pr} \Big{[} \mathcal{V} \mbox{ accepts } (G,\id,I) \mbox{ in every node given } \mathcal{M}\Big{]} \geq 1 - \eps.\]

\item \emph{ \soundness}
 If $(G,\id,I) \notin \cL$ then, for every prover $ {\mathcal M}$ such that 
 
 $({\mathcal V}, {\mathcal M}) \in {\dAM}_{\prot}[k,f(n)]$
(resp. $({\mathcal V}, {\mathcal M}) \in {\dMA}_{\prot}[k,f(n)]$),
\noindent
\[\mathbf{Pr} \Big{[} \mathcal{V} \mbox{ rejects } (G,\id,I) \mbox{ in at least one nodes given } \mathcal{M}\Big{]} \geq 1-\eps.\]
\end{itemize}
We also denote $\dAM[k,f(n)]= \dAM_{1/3}[k,f(n)]$ and $\dMA= \dMA_{1/3}[k,f(n)]$.
\end{definition}

We omit the subindex $\eps$ when its value is obvious from the context. 
For small values of $k$, instead of writing $\dAM[k,f(n)]$, we alternate \textsf{M}'s and \textsf{A}'s.
For instance: $ \dMAM[f(n)] = \dAM[3,f(n)] $. In particular $ \dAM[f(n)]= \dAM[2,f(n)] $, $ \dMA[f(n)] = \dMA[2,f(n)]$.

\begin{definition}
The shared randomness setting may be seen as if all the nodes, in any given round, sent the same random string to Merlin.
In order to distinguish between the settings of private randomness and shared randomness, we denote them by $\dAM^{\priv}[k,f(n)]$ and $\dAM^{\pub}[k,f(n)]$, respectively.
\end{definition}

In this paper, we are interested mainly in two languages, that we call \cograph\ and \distHereditary\, which are the languages of graphs that are cographs and distance-hereditary graphs, respectively.
Formally,

\begin{itemize}
	\item $\cograph= \{	\langle G, \id \rangle \text{ s.t. } G \text{ is a cograph.}\}$ 
	\item $\distHereditary=\{\langle G, \id \rangle \text{ s.t. } G \text{ is distance-hereditary.}\}$	
	%\item $\permutation = \{ \config \text{ s.t. }I: V \to [n]\text{ is injective.}\}$
\end{itemize}

Also, for a distributed language $\mathcal{L}$, the \emph{restriction of $\mathcal{L}$ to cographs}, denoted  $\mathcal{L}_\cograph$ is the subset of network configurations $(G, \id, I) \in \mathcal{L}$ such that $G$ is a cograph.

\section{Cographs }\label{sec:cograph}
In order to describe a protocol, we first show a way to distribute the proofs received by the network in such a way that we can centralize the verification process, by considering properties of cographs related to their connectivity.

\begin{lemma}\label{BCC}
	Given a connected $n$-node  cograph $G$, it is possible to construct a spanning tree $T$ of \emph{depth} two, such that each node at depth one has at most one child.
\end{lemma}
\begin{proof}
	As $G$ is connected, by definition it follows that $G$ can be obtained from the \emph{join} of two smaller cographs $G_1  $ and $G_2$. Then, let $G_1$ be the one with at least $\frac{n}{2}$ nodes.
	
	Consider now $\rho\in G_2$, the root of the $T$ to be constructed. It follows that $\rho$ has all of  $G_1$ as neighbors. Then, for each node in $G_1$, we set $\rho$ to be its parent in  $T$. 
	
	Finally, we have that the edges between $V(G_1)$ and $V(G_2-\rho)$ induce a complete bipartite graph $\hat{G}$, and the number of nodes in $G_2 - \rho$ is at most $n/2$. Therefore,  by Hall's theorem it follows that there exists a matching $M$ between both sides of $\hat{G}$ such that all nodes in $G_2-\rho$ have a match. Thus, for $u \in G_2-\rho$, we set its parent in $T$ to be its match $m(u) $ in $G_1$. The lemma follows.
\end{proof}

By the previous lemma, we know that for any two round protocol $\mathcal{P}$ over a cograph $G$ with cost $\Omega(\log n)$ bits, we may assume without loss of generality that there is a root $\rho$ with access to all coins and messages received by the whole network: Simply construct the a spanning tree given by Lemma \ref{BCC}, by choosing the root $\rho$ in a standardized manner: a bipartite graph can be easily verified with two colors, and $\rho$ can be chosen to be the node in $G_2$ with the smallest identifier. Then, it suffices to assign to each node $u$ of depth one in the spanning tree, both its proof and the proof received by its child $w$, along with the random coin it drew. Then, the nodes can locally verify the consistency of this message and the root will have received the entirety of messages in the network.
\begin{lemma}\label{lem:routing}
Given any $\dM$ (resp. $\dAM$) protocol with bandwidth $L$, we can construct a $\dM$ (resp. $\dAM$) protocol with bandwidth cost $L+ \cO(\log n)$ and where there exists a node $\rho$ which has access to all messages (resp. all messages and coins) in the network.
\end{lemma}
An advantage of this procedure is that we may simulate any protocol in the (non-deterministic) \textsf{One-Round Broadcast Congested Clique} model (by using the root $\rho$ as referee) by either using one round of interaction (if the simulated protocol deterministic) or two rounds (when the simulated protocol is randomized). From here it follows that we can use the public coin protocol by \cite{kari2015solving} to recognize cographs, therefore constructing a protocol for cograph detection in two rounds of interaction and $\cO(\log n)$ bits. That is, $\cograph\ \in \dAM^{\pub}[\log n]$.

For the sake of completeness, we now describe the protocol of \cite{kari2015solving}.

\begin{definition}\label{canonical Fam}
	Given a cograph $G=(V,E)$, we can define its \emph{canonical order} as follows:  
	
	We start by choosing the smallest pair of \emph{twins} (those with the smallest identifiers in lexicographic order) which we know to exist by Proposition \ref{cographCharac}. From there we choose and remove the smallest node from this pairing. Then, we repeat this process by finding another pair and removing one of its members until we end up with a single node.
\end{definition}

Let $p$ be a prime and $\phi = (\phi_w)_{w\in V}$ be a family of linearly independent polynomials in $\mathbb{Z}_p[x]$. Given $w\in V$ we define, $q_w = \sum_{w'\in N(w)} \phi_{w'}$ and $\bar{q}_w = q_w + \phi_w$. 

We also define the \emph{derivated polynomials} of $\phi$ as the collection 
\[\alpha_{u,v} = \phi_u-\phi_v\qquad \beta_{u,v} = q_u-q_v, \qquad \gamma_{u,v} = \bar{q}_u-\bar{q}_v, \quad u.v \in V\]
Now, given a pair of \emph{twins} $u$ and $v$, we assign to $G-v$ the pair of polynomials $\{\phi'_{w}\}_{w \in V-v}$ defined as 
\[\phi'_w= \begin{cases}
\phi_w  &\text{if } w \neq u\\
\phi'_u + \phi_v &\text{if } w=u
\end{cases}\]
	With this construction, from $\phi_u(x) = x^{\id(u)}$ it is possible to construct a sequence of polynomials  $\phi^i_u$ for $i \in [n]$ according to the \emph{canonical order} $\{v_i\}_{i=1}^n$ and  $u$ in the graph $G-\{v_j\}_{j=i+1}^n$. We call these functions the \emph{basic} polynomials of  $G$. And so the  \emph{canonical family} of polynomials of $G$ is defined as the union between its basic and derived polynomials. It follows that this family of functions has at most $3n^3$ elements.

\begin{definition}We say that a vector $m= ((a_w,b_w))_{w\in V} \in (\mathbb{{Z}}_p)^{2n}$ is \emph{valid for $G$ in $t\in \mathbb{Z}_p$} if there exists a family of linearly independent polynomials $(\phi_w)_{w\in V}$ in $\mathbb{Z}_p[X]$ such that $a_w = \phi_w(t)$ and $b_w=q_w(t)$ for each $w\in V$.
\end{definition}
\begin{lemma}\label{lem:canonTwin}
	Let  $m= ((a_w,b_w))_{w\in V} \in (\mathbb{{Z}}_p)^{2n}$ be a valid vector for $G$ in $t$. Consider $u,v$ to be a pair of twins in $G$ such that $a_u \neq a_v$. Then, the vector $m'=  ((a'_w,b'_w))_{w\in V-v} \in (\mathbb{{Z}}_p)^{2n-2}$ is valid for $G-v$ in $t$, where its coordinates are given by
	\[(a'_w, \: b'_w)=\begin{cases}
	(a_w, \:b_w) & \text{ if } w \in V- \{u,v\}\\
	(a_u + a_v, \: b_u-a_v\delta_{uv}) &\text{ if } w = u
	\end{cases}\]
	with $\delta_{uv}$ equals one if and only if $ a_u + b_u = a_v + b_v$
\end{lemma}
\begin{proof}
	Let $(\phi_w)_{w\in V}$ be a family of linearly independent polynomials associated to $m$. Given that $u$ are $v$ are twins and $a_u \neq a_v$ it follows that $a_u + b_u = a_v + b_v$ if and only if $u $ are $v$ adjacent. Therefore, $\delta_{u,v} = 1$ iff $u$ are $v$ are adjacent.
	
	Let now $\left(\phi'_w\right)_{w\in V-v}$ where $\phi'_w =\phi_w$ for all $w\neq u$, and $\phi_u' = \phi_u + \phi_v$. It is clear that this family is linearly independent. For $w\neq u $ we have that $a'_w = a_w = \phi_w(t) = \phi'_w(t)$. Also, $b'_w = b_w $ and $b_w = q_w(t)$. Now, as $u$  $v$ are {twins} either both nodes are in $N(w)$ or neither $u$ nor $v$ are. In both cases it follows that $b'_w = q'_w(t)$.
	
	By definition we have that $a'_u = a_u + a_v  = \phi_u(t) + \phi_v(t) = \phi'_u(t)$ and $b'_u = b_u - \delta_{uv} a_v$. As $b_u = q_w(t) = \delta_{uv} \phi_v(t)+ q'_w(t)$ , we finally have that $b'_u = q'_w(t)$.
\end{proof}

With this lemma now me can proceed to describe the protocol.

\begin{theorem}\label{thm:cograph}
	There is a distributed interactive proof with two rounds for the recognition of cographs graphs, i.e. $\emph{\cograph}\in \dAM[\log n]$. Moreover, the obtained protocol uses shared randomness and gives the correct answer with high probability.

	%For any $c>0$ we have that $\emph{\cograph} \in \dAM^{\pub}_{{1}/{n^c}}[\log n]$
\end{theorem}
\begin{proof}
	Let $G=(V,E)$ be an $n$-node graph. Without loss of generality we may assume the graph has identifiers in $[n]$ as, following Lemma \ref{BCC}, it is possible to implement a \permutation\ protocol in a single round:  Merlin sends to each node $v$ an identifier $\bar{\id}: V\to  [n]$ and the root, by receiving all proofs, can see that they all received distinct identifiers which are consistent with their original ones.
	
	Let $p$ be a prime such that $3n^{c+4}\leq p \leq 6n^{c+4}$. The protocol is the following: All nodes collectively generate a seed $t\in \mathbb{F}_p$ uniformly at random. Then Merlin sends to each node $w$  a message $m_w$ such that $m=(m_w)_{w\in V}$ is a valid vector for $G$ at $t$. Each node then computes such message by defining $\phi_w(x)=x^{\bar{\id}(w)}$.
	
	After the nodes exchange messages, following Lemma \ref{BCC} we obtain that the root $\rho$ owns a vector $m \in \mathbb{F}_p^{2n}$. From here, the root repeats the following procedure  at most $n-1$ times trying to construct a canonical ordering $\{v_i\}_{i=1}^n$ for $G$.
	
	At step $i$, it starts at graph $G^i$ and a vector $m_i\in \mathbb{F}_p^{2(n-i+1)}$  (where $G^1 = G$ and $m^1=m$) and looks for a pair of nodes $u,v$ in $G^i$ such that $a^i_u \neq a^i_v$ and either $b^i_u = b^i_v$ or $a^i_u + b^i_u = a^i_v +b^i_b$. 
	
	Then it chooses, among all pairs it has found, the first in lexicographic order. If no such pair exists, then he rejects. On the contrary, he defines $G^{i+1} = G^i-v$, and setting $v_{n-i+1} = v$ (without loss of generality we assume that $\id({v})<\id({u})$). Then the root computes $m^{i+1}$ from the previous vector $m^i$ following Lemma \ref{lem:canonTwin}. If the root reaches step $n-1$ then it accepts.

	\begin{itemize}
		\item {\completeness \& \soundness } It follows then that as the messages depend on the original identifiers and the root $\rho$ has access to all messages, then both acceptance errors depend solely in the \textsf{1BCC} construction. Now, by Lemma \ref{lem:canonTwin} it follows that the only point at which the protocol might fail is if the chosen $t$ turns out to be a root for any of the polynomials in the canonical family from Definition \ref{canonical Fam}. As there are at most $3n^3$ such polynomials, each of degree at most $n$ , we have that  the acceptance error is at most $3n^{4} / 3n^{c+4} = 1/n^c$ and the theorem follows.
	\end{itemize}
\end{proof}

As we mentioned in the introduction, the result given in \cite{kari2015solving} shows a stronger result. In fact, the referee not only can recognize a cograph but actually can \emph{reconstruct} it. In other words, when the input graph is a cograph, after the communication round the referee learns all the edges. In our context, this implies that the root $\rho$ not only recognizes cographs, but also can recognize any distributed language restricted to them.

\begin{theorem}
For every distributed language $\mathcal{L}$, there is a distributed interactive proof with two rounds for its restriction to cographs, i.e. $\mathcal{L}_\cograph \in \dAM[\log n]$. Moreover, the protocol uses shared randomness and gives the correct answer with high probability.

\end{theorem}

\begin{proof}
	It is sufficient to notice that the tree-root $\rho$ in the construction from Lemma~\ref{BCC} has access to all proofs in the network. In particular, the \id's and positions for each node in the twin-ordering $\pi$. As such, $\rho$ has knowledge of the entire topology of the network and its inputs (provided that these are of size $\cO(\log n)$) and can compute any property related to them, with the acceptance error matching that of the verification procedure in Theorem~\ref{thm:cograph}. As for the rest of the nodes, they simply accept and delegate this decision to the root.
%	\TODO{explicar aqui como el referee va recuperando las aristas del grafo. En el paper JCSS esta explicado para grafos de ancho modular acotado (los cografos tienen ancho 2). Tambien podriamos sacarnos el pillo citando ese paper y diciendo que ahi se explica como reconstruir. Lo que si agregaria es como se decide un predicado cualquiera (que nadie hace nada mas salvo la raiz que mira el grafo y decide)}
%	\TODO{No debería ser eso un teorema aparte?, o agregarlo al teorema principal como un "moreover"
%	Otra cosa: el hecho de que un nodo de forma centralizada identifique todo el orden, debería bastar para poder reconstruir todo o no? no es necesario usar lo del paper de reconstrucción.}
\end{proof}

\section{Distance Hereditary Graphs}\label{sec:dist_hed}
	Following the protocol described for cographs, it is possible to derive an interactive protocol for distance-hereditary graphs, which admit a similar construction. Indeed, as described before, any distance-hereditary graph can be constructed by sequentially adding twins or pending nodes. 
Notice that for the protocol in Theorem~\ref{thm:cograph}, the verification process is done by the root as it prunes the graph for $n-1$ steps. This leads to an order by which the nodes were selected, and we call it \emph{canonical} ordering. While we can not delegate the verification routine to a single node (as distance-hereditary graphs can have arbitrarily large diameter),  we can distribute the verification process by letting different nodes check different steps of the computation. As  the rule described in Lemma~\ref{lem:canonTwin} for pruning the graph involves only the pair of twins at each step, we only need to find nodes that, for a fixed node $v$, can receive all the proofs sent by $v$, its twins and its pending nodes.

First, in order to prune the graph in this new setting, we need a rule for pruning pending nodes from a graph and updating the vectors of each node accordingly. Here, we use the definition of a \emph{valid} vector as described in Section~\ref{sec:cograph}.
\begin{lemma}\label{lem:disthed}
	Let  $m= ((a_w,b_w))_{w\in V} \in (\mathbb{{Z}}_p)^{2n}$ be a valid vector for $G$ at some point $t$. If $u\in G$ has $v$ as a pending node adjacent to it, then, the vector $m'=  ((a'_w,b'_w))_{w\in V-v} \in (\mathbb{{Z}}_p)^{2n-2}$ is valid for $G-v$ in $t$, where the coordinates of $m'$ are given by
	\[(a'_w, \: b'_w)=\begin{cases}
	(a_w, \:b_w) & \text{ if } w \in V- \{u,v\}\\
	(a_w , \: b_w-a_v) &\text{ if } w = u
	\end{cases}\]
	
\end{lemma}
\begin{proof}
	If $\{\phi_w\}_{w\in V}$ is a family of linearly independent polynomials associated to $m$, we can use this family and, as $v$ was only connected to $u$, it follows that $b'_u = b_u - a_v = q_u(t)-\phi_v(t) = \sum_{w\in N(u)-v} \phi_w(t) = q'_u(t)$.
\end{proof}
 
In order to distribute the verification procedure, for any fixed $v$ we wish to set a node to compute the correctness of the vectors of all nodes assigned as twins of $v$. Indeed, for a fixed ordering $\pi$ for pruning the graph and a node $v$ with $\pi_v<n$, consider the \emph{predecessor} of $v$, denoted by $\ant(v)$, to be $v$'s neighbor whose value for $\pi(\cdot)$ is immediately after that of $v$ among its neighbors. As all previous nodes in the order which are twins of $v$ have the same neighborhood, it follows that all these nodes must be adjacent to $\ant(v)$. In case that no such a node exists, by assuming that $G$ is connected, it follows that the last node according to $\pi$ which is assigned as a twin of $v$ must be a \emph{true} twin and, therefore, be adjacent to him. And the same reasoning holds.
 
 In our protocol, the prover will provide the nodes with the ordering in which they will be pruned. Naturally, this ordering has to be verified in order to preserve the soundness of the protocol. The language \permutation\ refers to the one in which each node $v$ of a graph owns an input $a_v\in [n]$, where all inputs are distinct. Formally: 
 	 \[\permutation = \{ \config \text{ s.t. }I: V \to [n]\text{ is injective.}\}\]

This problem was addressed in~\cite{naor2020power}, where the authors provide a two-round distributed interactive proof recognizing the language with high probability and using shared randomness. %\TODO{dar algunos detalles del protocolo?}

\begin{proposition}[\cite{naor2020power}]
\label{prop:perm}
	There is a two-round distributed interactive proof  recognizing the language \emph{\permutation}, i.e. $\emph{\permutation}\ \in \dAM[\log n]$. The protocol uses shared randomness and gives the correct answer with high probability.

\end{proposition}

 Thus, the main strategy of our protocol is that, given an initial vector $(a_v, b_v)$ for a node $v$ in the graph, each node $\ant(v)$ has the task of updating this vector until it obtains the correct vector $v$ should have at the time he is pruned from the graph, which we denote by $(a^\pi_v, b^\pi_v)$, for this, each node $u$ which is a twin of $v$ provides its vector $(a^\pi_u, b^\pi_u)$ (which is proved to be correct by some other node) and so the predecessor of $v$ compares and updates $v$'s vector according to the rules from Lemmas~\ref{lem:canonTwin} and~\ref{lem:disthed}. %With this,  we deduce a protocol.

\begin{theorem}\label{thm:distHed}
	There is a distributed interactive proof with three rounds for the recognition of distance-hereditary graphs, i.e. $\emph{\distHereditary}\in \dMAM[\log n]$. Moreover, the protocol uses shared randomness and gives the correct answer with high probability.
\end{theorem}
\begin{proof}
	We need to check that, given an ordering $\pi$,  we can compute the coordinates for a valid vector at each step of the computation by delegating this information to the correct nodes. For each node $v$, we define $\textsf{twin}(v)$ to be the node assigned as its twin, $\ant(v)$ to be its predecessor according to the previous definition, $\textsf{Twins}(v)$ to be the set of nodes $u$ such that $\textsf{twin}(u)=v$ and $\textsf{Pending}(v)$ to be the set of nodes that have $v$ as their unique neighbor according to $\pi$.

	First, Merlin sends, to each node $v$
	
	\begin{enumerate}
		\item Its position in the ordering $\pi$ given by $\pi_v$. Let us call $N^*(v)$ the set of neighbors of $v$ that are after it in $\pi$. Formally, $N^*(v) = N(v) \cap \{w: \pi_w > \pi_v\}$.
		\item The \id\ of the unique neighbor which is immediately after it at $\pi$, which we denoted by $\ant(v)$. Formally, 
		$\ant(v) = \textrm{argmin}\{\pi_u: u \in N^*(v)\}$
	
		\item If $v$ is removed from the graph as a \emph{pending node}
		\begin{enumerate}
			\item The \id\ of the node to which $v$ is a pending node, denoted by $\textsf{pending}(v)$. Formally,
			
			$\textsf{pending}(v) = u$ if and only if $N^*(v) = \{u\}$. 
			
		\end{enumerate}
		\item If $v$ is removed from the graph as a \emph{twin},
		\begin{enumerate}
			\item  The $\id$ of $v$'s twin at step $\pi_v$ of the computation, denoted by $\textsf{twin}({v})$. Formally, 
			$\textsf{twin}(v) = \textrm{argmin}\{\pi_u : N^*(v) = N(u)\cap \{w: \pi_w > \pi_v \}\text{ and } \pi_u >\pi_v\}$.

			\item The \id\ of $\ant(\textsf{twin}(v))$, i.e. the \id\ of the predecessor of the twin of $v$ according to $\pi$.
			\item A single bit, indicating whether $\textsf{twin}(v)$ it is a \emph{true} or \emph{false} twin of $v$.
		\end{enumerate}
		
		\item The number of neighbors $u$ of $v$ such that $\textsf{twin}(u) = v$ according to $\pi$, denoted by $|\textsf{Twins}(v)| = |\{u : \textsf{twin}(u) = v\}|$.
	\end{enumerate}
	Unfortunately, this set of certificates are not sufficient. By receiving them, $\ant(v)$ can collect all proofs received by each node $u$ with $\textsf{twin}(u)=v$ yet we cannot follow the same decomposition in order to check that the pair $(a_v^\pi, b_v^\pi)$ is correct. Indeed, the nodes in $\textsf{Pending}(v) = \{w: \textsf{pending}(w)=v\}$ are not adjacent to $\ant(v)$, so $\ant(v)$ can not see these proofs.
	
	To fix this, Merlin will distribute a set of proofs between $v$ and those in $\textsf{Twins}(v)$. First, Merlin sends to $v$:
	\begin{enumerate}
		\setcounter{enumi}{5}
		\item The number of nodes $u$ such that $\textsf{pending}(u) = v$, that is, $|\textsf{Pending}(v)|$.
		\item The \id\ of the node $u\in \textsf{Pending}(v)$ with the smallest value for $\pi$, denoted by $\textsf{m-leaf}(v)$, as well as $\pi_u$. Formally, $\textsf{m-leaf}(v) = \textrm{argmin}\{\pi_u : u \in \textsf{Pending}(v)\}$
	\end{enumerate}
	Now, fix some node $u$ such that $\textsf{twin}(u) = v$. Then, Merlin sends to $u$
	\begin{enumerate}
		\setcounter{enumi}{7}
		\item The \id\ of the node in $\textsf{Twins}(v)$ whose value for $\pi$ is immediately after $\pi_u$, denoted by $\textsf{co-twin}(u)$. Formally $\textsf{co-twin}(u) = \textrm{argmin}\{\pi_w : w \in \textsf{Twins}(v)\text{ and } \pi_w>\pi_u \}$
		\item The number of nodes in $\textsf{Pending}(v)$  with values for $\pi$ between those for $u$ and $\textsf{co-twin}(u)$.
	\end{enumerate}
	
	% TODO: \usepackage{graphicx} required
	\begin{figure}[!ht]
		\centering
		\includegraphics[width=0.85\linewidth]{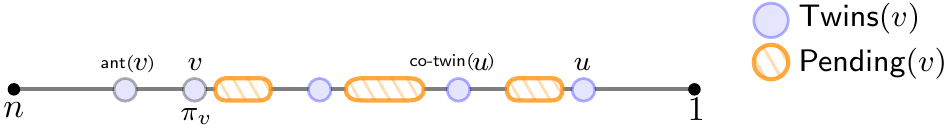}
		\caption{Visualization of $\ant(v)$ and the sets $\textsf{Pending}(v)$ and $(v)$ for  some node $v$: the neighbor $\ant(v)$ sees all nodes in $(v)$, which are the nodes $u$ that are assigned as twins of $v$. If we order these according to $\pi$, there may have nodes in $\textsf{Pending}(v)$ positioned between them.}
		\label{fig:disthered}
	\end{figure}
	After this first interaction, all nodes (collectively)  send a seed $t\in \mathbb{F}_p,$ with $p= \poly(n)$ a prime number to be defined accordingly. From here, Merlin answers to each node $v$ with the following set of messages. First, he sends the pair of coordinates $(a_v,b_v)$ which belong to $v$ at the start of the computation, as well as the pair he would have sent at step $\pi_v$ of the computation, denoted by $(a_v^\pi,b_v^\pi)$. Then, he sends the value $P(v,t) = \sum_{u\in \textsf{Pending}(v)}\phi_u(t)$  which is an encoding for the set of $v$'s pending nodes according to $t$. Finally, if $v$ is assigned as a twin, with $\textsf{twin}(v)=u$,  Merlin sends the message $\sum_{w\in S(v)}a_w$, where the set $S(v)$ is defined as
	\[ S(v)= \{ w \in \textsf{Pending}(u):    \pi_{\textsf{co-twin}(v)} > \pi _w > \pi_v\}  ,\]
	that is, the set of pending nodes connected to $v$'s twin such that they appear between $\textsf{twin}(v)$ and $\textsf{co-twin}(v)$.

	Finally, at the verification round, the nodes exchange their certificates and they try to collectively compute the correct values for $(a_v^\pi,b_v^\pi)$ as follows: 
	First, each node checks that its original values for the pair $(a_v,b_v)$ are correct, along with the size of its set of pending nodes $\textsf{Pending}(v)$ and  $\ant(v)$. Now, if $v$ is assigned as a leaf, it simply checks that it is adjacent to $\textsf{pending}(v)$ and that it is the only neighbor with a larger position at $\pi$. If $v$ is assigned as a twin instead, with $u = \textsf{twin}(v)$, it delegates this verification to the node $\ant(u)$ as follows.

	\begin{figure}
		\centering
		\includegraphics[width=0.85\linewidth]{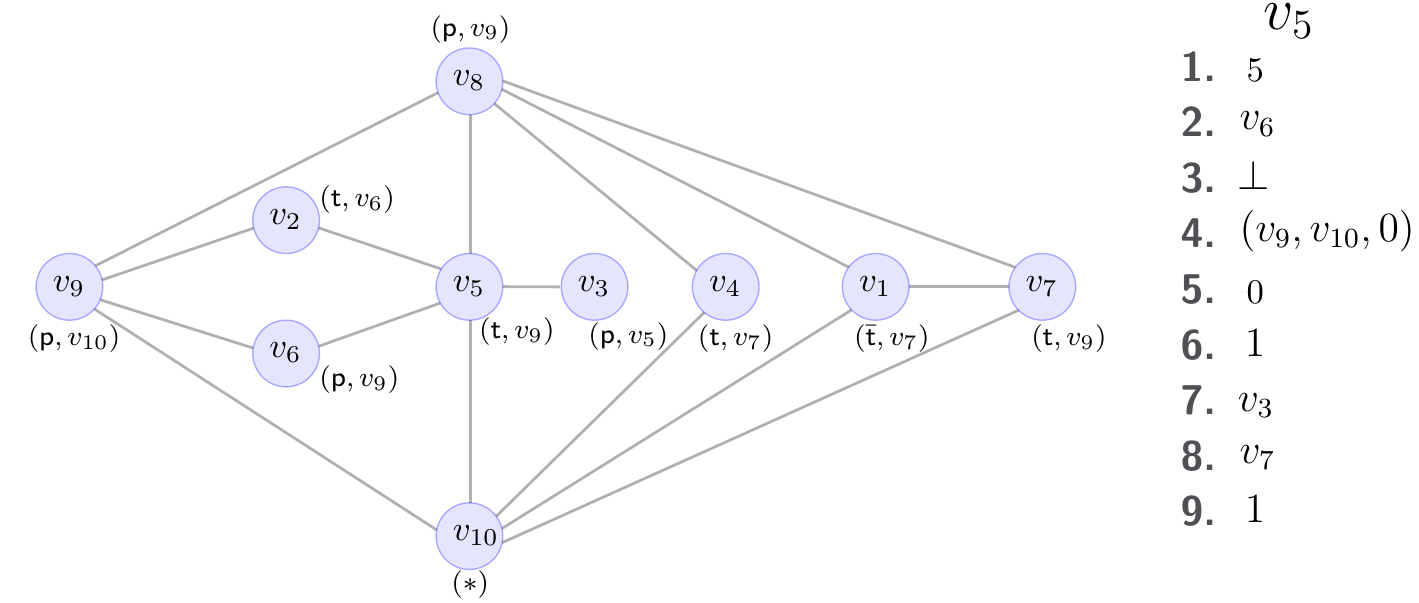}
		\caption{An example of a first round of interaction for a fixed node $v_5$. Given a distance-hereditary graph, each node is identified according to their ordering $\pi$ and labeled depending on the way they are removed from the graph. The first entry shows whether some node is a \emph{false twin} (\textsf{t}), a \emph{true twin} ($\bar{\textsf{t}}$) or a \emph{pending node} (\textsf{p}), with the second entry showing its assigned node on each of these cases. A list of all certificates sent by Merlin to $v_5$ according to the previous list is included.}
		\label{fig:distheredprot}
	\end{figure}
	
	First, the node $\ant(u)$ compares the number of nodes adjacent to it which are assigned as twins of $u$ with the size of the set $\textsf{Twins}(u)$ sent by Merlin. Then, it sorts the nodes in $\textsf{Twins}(u)$ in increasing order according to $\pi$ , and name them $w_1, \dots w_{|\textsf{Twins}(v)|}$. For each $i$, it checks that $\textsf{co-twin}(w_{i}) = w_{i+1}$. The only issue we are left to determine is the values for $P(i,t)= P(w_i,t)$. That is, the encoding for all pending nodes that lie between $w_i$ and $w_{i+1}$ for $i= 1, \dots \textsf{Twins}(u)$, where we set $w_{\textsf{Twins}(u)+1}= u$. We can assume that we know the encoding for any pending node that occurs before $w_1$ as they can be checked by $u$ and sent to $\ant(u)$ during the verification round. This is because Merlin can send to node $u$ the position of the first node which has $u$ as a twin. 
	
	To obtain $P(i,t)$, we simply need to show that $\ant(u)$ is able to partition the set $\textsf{Pending}(u)$ (and its encoding) into groups according to their positions in the permutation in-between nodes assigned as twins. Indeed, as $\ant(u)$ knows the size of $\textsf{Pending}(u)$ , it simply needs to sum the number of pending nodes that appear between $w_i$ and $w_{i+1}$ according to $w_i$ from message \lipics{(9)}. We know that these value are correct. Each $w_i$ can count them (it knows that $w_{i+1} = \textsf{co-twin}(v)$). Thus, $\ant(u)$ computes the size of the parts of $\textsf{Pending}(u)$ according to the collection $\{w_i\}_{i}$. If these values do not match the size of $\textsf{Pending}(u)$, it simply rejects. 
	
	Finally, as $\ant(u)$ knows that the partition is correct, it simply considers the values for $P(i,t)$ provided by $w_i$ and computes the values for $(a^i_v,b^i_v)$ at each step of the computation: At step $i$, it obtains a pair $(\bar{a}_v^i, \bar{b}_v^i)$, if $w_i$ is a \emph{true} twin, it  compares its sum with that of $(a_{w_i}^{\pi}b^{\pi}_{w_i})$ and rejects in case these values are not equal. Otherwise, it simply compares the values for $\bar{b}_v^i$ with $b^{\pi}_{w_i}$ and reject if these are not equal.
	
	 Then, in order to obtain $(a_v^{i+1},b_v^{i+1})$ it follows the rule from Lemma~\ref{lem:canonTwin}, deletes $P_i(t)$ from $b^{i+1}_v$ and goes to the next step. At the end of this computation, it remains to check whether the obtained pair equals $(a_v^{\pi}, b_v^{\pi})$, and accept or reject accordingly.	
	
	\begin{itemize}
		\item \completeness\ An honest prover will provide the nodes with the correct ordering. Then, as we have described above, the  nodes compute the correct values for each of their coordinates in a valid vector.
		\item \soundness\ If $G$ is not distance hereditary, because Proposition \ref{prop:perm}, we can assume that Merlin provides an ordering that corresponds to a permutation with high probability. Therefore, it remains to check that $\pi$ satisfies the above properties. Indeed, each $v$ which is assigned as a leaf, trivially computes it has a unique neighbor with a higher value for $\pi$ that matches $\textsf{pending}(v)$ and therefore it accepts. For any node $v$, by assuming that for nodes $w$ removed at previous steps the vectors $(a_w^\pi, b_w^\pi)$ are correct,  $\ant(v)$ should compute that the pair sent to $v$ is correct at each step of the computation and that all nodes $u$ such that $\textsf{twin}(u)=v$ are correct. We have that the values for the set of pending nodes in-between nodes assigned as twins is correctly computed as described above. Therefore, by an inductive proof all vectors should be correctly computed. Moreover, for a valid pair which is not correct at some step of the computation, it should occur that the seed $t$ is bad for some canonical polynomial. As this family of functions is polynomially bounded beforehand, we have that at least one node should reject with high probability.
	\end{itemize}
	
\end{proof}

\section{Lower Bounds}\label{sec:lower}

In this section, we provide lower-bounds on the certificate size of distributed interactive proofs for \emph{\cograph}\ or \emph{\distHereditary}. 

\begin{theorem}\label{thm:glueing}
	If \emph{\cograph}\ or \emph{\distHereditary} is in $\dM[f(n)]$, then $f(n)=\Omega(\log n)$. Moreover if, for any fixed $k$,  $\emph{\cograph}$ or $\emph{\distHereditary} \in \dAM^{\pub}[k, g(n)]$, then $g(n)=\Omega(\log\log n)$.
\end{theorem}
\begin{proof}
	We first describe a construction for lower bounds in $\dAM^{\pub}$, then we explain how to deduce stronger lower bound on $\dM$ protocols. 

	 Without loss of generality, we may assume that  $n$ is even. Let  $A$ to be a partition of $[1,n^2]$ into $n$ sets of size $n$, and set $B, C, D$ to be three disjoint sets in $[n^2+1 , 2n^2]$ of size $n$. Let $\mathcal{G}$ be a family of $n$-node cographs.
	 
	 Set $\mathcal{G}_A $  to be the set of labeled graphs in $\mathcal{G}$, with label sets picked from $A$. Let $F_a$ be a graph in $\mathcal{G}_A$
	 
	 For $(F_a, b,c,d) \in \mathcal{G}_A\times B\times C\times D$ let $G(F_a,b,c,d)$ be the graph defined by the disjoint union of graphs  $F_a$ and a triangle formed by $(b, c, d)$ plus a  join between these and an additional node $x_a$. The node $x_a$ is labeled with a different number in a set from $ [2n^2 +1, 3n^2]$ disjoint from $B, C$ and $D$.  Observe  that all nodes in $F_a$ communicate with $(b, c, d)$ only through the node $ x_a$

	\begin{figure}[!ht]
		\centering
		
		\includegraphics[width=0.4\linewidth]{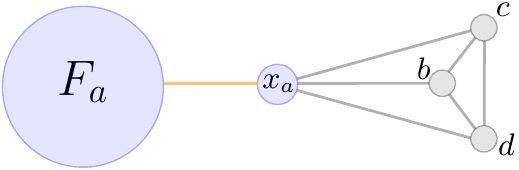}
		\caption{The auxiliary graph $G(F_a,b,c,d)$, with $(a, b, c, d)\in A\times B\times C\times D$. The orange edges represent a \emph{join} between two sets of nodes.  This is a yes-instance for \cograph, as it corresponds to a join between a disjoint union between a triangle $(b, c, d)$ and a cograph $F_a$, with the node $x_a$. As cographs are a special case of distance hereditary graphs, this implies that it is a Yes-instance for both.}
		\label{fig:lowerBoundMake}
	\end{figure}
	
	Let $\mathcal{P}$ be a $k$-round distributed interactive proof with shared randomness verifying a property $\mathrm{P}\in \{\cograph, \distHereditary\}$ with bandwidth $K = \delta \log \log n$ and error probability $\eps$. 
		
%	Let us call $x_a$ the \emph{core} of $G(F_a, b, c, d)$. We can assume, without loss of generality, that $\mathcal{P}$ is \emph{simple}, that is, it satisfies that for any (shared) random string generated by Arthur, the nodes in the core $\{x_A,y_A,z_A\}$ receive the same proof. Indeed, if we have a protocol that is not simple, with cost  $L$, we can design a new protocol that is simple and whose proof has length $3L$ by making each node pick their portion of the proof and going by the original protocol afterwards.
		
	Given a sequence of random strings $r=(r_1,r_2,\dots r_k)$, we call $m^r$ the sequence indexed by vertices $v\in V(G[F_a,b,c,d])$, such that $m^r_v$ is the set of certificates that Merlin sends to node $v$ in protocol $\mathcal{P}$, when Arthur communicates string $r_i$ on round~$i$. 
	Let $m_{abcd}: \{0,1\}^{Kk} \to \{0,1\}^{Kk}$ be the function that associates to each sequence $r=(r_1,r_2,\dots r_k)$ the tuple $( m^r_{x_a}, m^r_{b},m^r_{c}, m^r_{d},)$ such that it extends to a proof assignment for the nodes in $F_{a}$  that makes them accept whenever the $x_a$ accepts.
		
	Now consider the complete 4-partite, 4-uniform hyper-graph graph $\tilde{G} = A \cup B\cup C\cup D$. For each $a\in A, b\in B, c\in C$ and $d\in D$, color the edge $\{a,b,c,d\}$ with function $m_{abcd}$. There are at most $2^{Kk2^{Kk}}$ possible functions. Therefore, by the pigeonhole principle, there exists a monochromatic set of hyper-edges $H$ of size at least $\frac{n^4}{2^{Kk2^{Kk}}}$.

	Observe that  for sufficiently small $\delta$ and large $n$, $2^{Kk2^{Kk}} = (\log n)^{\delta k \log^{\delta k} (n)} = o(n^{1/{8}})$. Indeed, if $n > 2^{k\delta}$ and $\delta < 1/(2^{4}k)$  have that $ \delta k\log^{\delta k}(n) \log \log (n )\leq \log^{2\delta k}(n) <\frac{1}{8} \log n$.  Now, following a result by Erd\H{o}s~\cite{erdos1964extremal}, by setting $\ell =2, r=4$ we have that there exists a $K^{(r)}(\ell)$ subgraph in $\tilde{G}$ induced by $H$. That is, the complete $r$-uniform, $r$-partite hyper-graph, where each part has size exactly $\ell$. Let $\{a_i,b_i, c_i, d_i\}_{i=1}^2$ be the nodes involved in such a graph.

	Consider now the graph $G(a_i, b_i, c_i, d_i)$ defined as follows: First. take a disjoint union of the graphs $F_{a_1}$ and $F_{a_2}$. Then, for each $i \in \{1, 2\}$ add nodes $x_{a_{i}},  b_{i}, c_{i}$ and $ d_{i}$, labeled with different labels in $[2n^2+1, 3n^2]$  correspondent to the yes instances formed by the graphs $F_{a_i}$, and the nodes $(x_{a_i}$, $b_j$, $c_k$ and $d_h$ . 	For each $i\in \{1,2\}$, the graph $F_{a_i}$ is \emph{joined} to the node $x_{a_i}$ which, in turn, is adjacent to the nodes $b_i, c_i$ and $d_{i+1}$.  Also, the node $b_i$ is adjacent to the node $c_{i+1}$, the node $c_{i}$ is adjacent to the node $d_{i+1}$ and, finally, the node $d_i$ is adjacent to the node $b_{i+1}$. Where, in all these cases, the $i+1$ is taken$\mod 2$.	
	
	Now, we show that the graph $G(a_i, b_i, c_i, d_i)$ is a No-instance for the properties $\cograph$ and $\distHereditary$. Indeed, as our construction contains an induced 6-cycle it follows that it has an induced $P_4$ which make it a No-instance for \cograph\ as it is the class of $P_4$-free graphs. As for the problem \distHereditary, it suffices to consider another characterization for this class, namely these are the graphs such that any cycle of length at least 5 has a pair of crossing diagonals~\cite{brandstadt1999graph}. It follows that a graph with an induced 6-cycle can not be distance-hereditary.
	It remains to check that this No-instance is capable of fooling the verifier. Indeed, the nodes $x_{a_{1}}, x_{a_{2}}$ receive the same answers by Merlin which extend to assignments for the nodes in $F_{a_1}$ and $F_{a_2}$,  making them accept with the same probability as these nodes. As they locally place themselves in a previously defined yes instance, all vertices accept two thirds of all possible random coins. This contradicts the fact that $\mathcal{P}$ was a correct distributed interactive proof for $\mathrm{P}\in \{\cograph, \distHereditary\}$.	
	\begin{figure}
		\centering
		\includegraphics[width=0.8\linewidth]{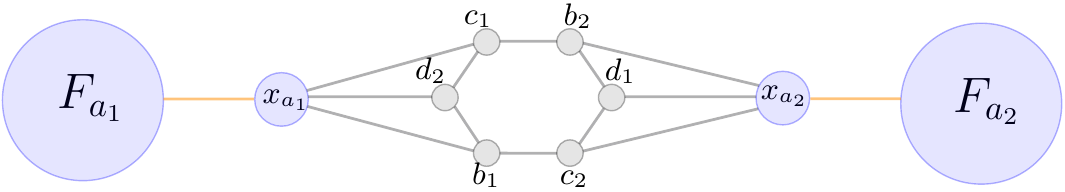}
		\caption{ A No-instance for \cograph\ and \distHereditary. By containing a 6-cycle formed by the nodes $b_i, c_i$ and $d_i$ it has two induced paths of different length between the nodes $c_2$ and $d_2$ and, therefore, can not be distance-hereditary nor a cograph.}
	\end{figure}
\end{proof}

Is important to mention that this technique can not provide lower bounds stronger than $\Omega(\log n)$ for the $\dM$ class or $\Omega (\log\log n)$ for the $\dAM[k]$ class as a natural barrier in this construction is that each configuration admits identifiers in a polynomially bounded set, and therefore the previously defined sets $A$ and $B$ can not be too large.
If we recycle the graph constructions from Theorem~\ref{thm:glueing} and follow a result by Fraigniaud et al.~\cite{fraigniaud2019distributed}, we can also obtain ``strong'' lower bounds for the models \dMA\ in both its shared and private randomness versions if the acceptance probability is sufficiently high.
\begin{corollary}\label{cor:lower}
	If any of the  problems \emph{\cograph}\ or \emph{\distHereditary}\ is in $\dMA_{1/7}[f(n)]$, then $f(n) = \Omega(\log n)$. 
\end{corollary}

\begin{proof}
	We use the constructions from Theorem~\ref{thm:glueing} and change the ways each edge in the auxiliary (hyper) graph is colored. By bounding the acceptance probability according to what each node in the construction can see, we obtain the result.

	Suppose we have a protocol $\mathcal{P}$ in \dMA\ with acceptance probability $\eps\leq \frac17$. We repeat again the construction $G(F_a, b, c, d)$ with $(a, b, c, d)\in A\times B\times C\times D$, as well as the 4-partite hyper-graph $\tilde{G}= A\cup B\cup C \cup D$. Using the result by Erd\H{o}s~\cite{erdos1964extremal} we obtain a monochromatic $K^{(4)}(2)$ graph colored by the proof $\bar{m}$. Now in order to study the acceptance probability of $G(a_i, b_i, c_i, d_i)$, we define $G_{ijkh}$ to be the graph $G(a_i, b_j, c_k, d_h)$, from here we bound the acceptance probability as
	
	\begin{align*} \mathbf{Pr}[\text{Some  }v\text{ in  }V \text{ rejects}]&\leq \mathbf{Pr}[\text{Some }v\text{ in  }V(F_{a_1})\cup\{d_2\} \text{ rejects at  }G_{1112}]\\
	&+\mathbf{Pr}[\text{Some }v\text{ in  }V(F_{a_2})\cup\{d_2\} \text{ rejects at  }G_{2221}]\\
	&+\mathbf{Pr}[\text{Some }v\text{ in  } \{c_1\} \text{ rejects at  }G_{1212}]\\
	&+\mathbf{Pr}[\text{Some }v\text{ in  }\{c_2\} \text{ rejects at  }G_{2121}]\\
	&+\mathbf{Pr}[\text{Some }v\text{ in  }\{b_1\} \text{ rejects at  }G_{1122}]\\
	&+\mathbf{Pr}[\text{Some }v\text{ in  }\{b_2\} \text{ rejects at  }G_{2211}]\\
	&< \frac67\end{align*}
	Once again each term represents a portion of the graph that accepts with good probability (at least $6/7$) as any combination of the above vicinities is considered in the monochromatic $K^{(r)}(\ell)$. Finally, we have a contradiction on $\mathcal{P}$'s correctness as $G(a_i,b_i,c_i,d_i)$ is a  bad instance that should accept with probability smaller than $1/7$. 
\end{proof}

\section{Discussion}

We have shown that cographs and distance-hereditary graphs can be recognized with two-round distributed interactive protocols 
using  logarithmic certificates and shared randomness.  A natural step forward is to improve our results by using fewer rounds or smaller certificates. Interestingly, our lower-bounds show that the two types of improvements are not possible to achieve simultaneously. Indeed, we have shown that, even with a fixed number of rounds  (\dM\ (PLS) or \dMA)   recognizing these classes require certificates of size at least logarithmic.

Another interesting question is the existence of a non-trivial \dM\ protocol for \cograph. In fact, we have not been able to provide a sub-linear one-round interactive potocol for this problem (the fact that  $\cograph\in \dM[n] $ is trivial). On the other hand, from Lemma~\ref{BCC}, 
it follows that any one-round deterministic protocol in the \textsf{1BCC}  model recognizing cographs, would immediately imply a \dM\ protocol for \cograph. 

Currently, it is not  known whether
recognizing cographs can be done through a deterministic \textsf{1BCC} protocol. Yet, we take note  of the following corollary of Lemma~\ref{lem:routing} which may be of interest, in case that there exists a positive result in the future.

\begin{corollary}
	Any non-deterministic protocol for \emph{\cograph}\ in the \emph{\textsf{1BCC}} model with bandwidth $L$ would imply a \dM\ protocol with bandwidth $L+ \cO(\log n)$.
\end{corollary}
This adds a new perspective either for the search of protocols (simulating a protocol in the \textsf{1BCC} model following Lemma~\ref{BCC}) or searching for super-logarithmic lower-bounds for one-round interactive proofs, implying lower bounds in the \textsf{1BCC} model.

\bibliography{i-ref}

\end{document}